\documentclass{article}
\usepackage[utf8]{inputenc}
\usepackage{amsmath}
\usepackage{amsthm}
\usepackage{lscape}
\usepackage{amsfonts}
\usepackage{mathtools}
\usepackage{amsmath,amssymb}
\usepackage{graphicx}
\usepackage[vcentermath]{youngtab}
\usepackage{mathtools}
\usepackage{latexsym,amssymb,amsmath,amsfonts, amsthm, xcolor}

\newtheorem{theorem}{Theorem}
\newtheorem{lemma}[theorem]{Lemma}

\newtheorem{proposition}[theorem]{Proposition}

\title{Controllability of the  Periodic Quantum Ising Spin Chain}

\author{Domenico D'Alessandro\thanks{Department of Mathematics, Iowa State University, Ames, IA 50011, daless@iastate.edu},  and Yasemin Isik\thanks{Department of Mathematics, Iowa State University, Ames, IA 50011, yisik@iastate.edu} }

\begin{document}

\maketitle

\begin{abstract}

The Ising spin model is a fundamental theoretical model in physics which has received large attention in the context of quantum phase transitions  
\cite{QPT}. The Hamiltonian of the system consists of two parts whose ground states describe   different macroscopic observable behavior. A parameter that interpolates between such Hamiltonians may physically represent  a transverse magnetic field which drives the phase transition between  two macroscopic states. Such a model   has recently been applied to perform quantum annealing in quantum information \cite{DanielAN}, a process that can be seen as a controlled quantum dynamics protocol.

In this paper, we present a controllability analysis of the quantum Ising  periodic  chain  of $n$ spin $\frac{1}{2}$ particles  \cite{Pierotti} where the interpolating parameter plays the role of the control.
A fundamental result in the control theory of quantum systems states that the set of achievable  evolutions  is (dense in) the Lie group 
corresponding to the Lie algebra generated by the Hamiltonians of the system. Such a {\it dynamical Lie algebra} therefore characterizes all the state transitions available for 
a given system.  For the Ising spin  periodic chain we characterize such a dynamical Lie algebra and therefore the set of all reachable states. 
In particular, we prove that the dynamical Lie algebra is  a  $(3n-1)$-dimensional Lie sub-algebra of $su(2^n)$ which is a direct sum of a two dimensional center and a $(3n-3)$-dimensional semisimple Lie subalgebra. This in turn   is the direct sum of $n-1$ Lie algebras isomorphic to $su(2)$ parametrized by the eigenvalues of a fixed matrix. We display the basis for each of these Lie subalgebras.  Therefore the problem of control for the Ising spin periodic chain  is, modulo the two dimensional center, a problem of simultaneous control of $n-1$ spin $\frac{1}{2}$ particles.   In the process of proving this result, we develop some tools which are of general interest for the controllability analysis of quantum systems with symmetry.

\end{abstract}

\vspace{0.5cm}


\vspace{0.5cm}

\section{Introduction}\label{IandM}

We shall study  the controllability of a system of $n$ spin $\frac{1}{2}$ particles in a periodic chain with next neighbor Ising interaction. The Hamiltonian of the system is  of the form 
\begin{equation}\label{Hamiltonian}
H=H_0(1-u)+H_1u, 
\end{equation}
with 
\begin{equation}\label{H0H1}
H_0:=\sum_{j=1}^{n-1} \sigma_z^j \sigma_z^{j+1}+ \sigma_z^n \sigma_z^1, \qquad H_1 =\sum_{j=1}^n \sigma_x^j.
\end{equation}
Here we use the standard physics convention  of denoting by $\sigma^j$ the tensor product of $n$,  $2 \times 2$,  identity operators, $\bf 1$, except in position $j$ which is occupied by $\sigma$. The matrices $\sigma_{x,y,z}$ are the {\it Pauli matrices},    
\begin{equation}\label{PauliMat}
\sigma_x:=\begin{pmatrix} 0 & 1 \cr 1 & 0 \end{pmatrix}, \qquad \sigma_y:=\begin{pmatrix} 0 & i \cr -i & 0 \end{pmatrix}, \quad \sigma_z:=\begin{pmatrix} 1 & 0 \cr 0 & -1 \end{pmatrix}, 
\end{equation}
which satisfy the {\it commutation relations}, 
\begin{equation}\label{commurel}
[i\sigma_x, i\sigma_y ]=2i\sigma_z, \qquad [i \sigma_y, i \sigma_z] =2i \sigma_x, \qquad [i\sigma_z, i \sigma_x] =2i\sigma_y,  
\end{equation}
and the  {\it anti-commutation relations}, 
\begin{equation}\label{anticommurel}
\{\sigma_x, \sigma_x\}=\{ \sigma_y, \sigma_y\}=\{\sigma_z, \sigma_z\}=2{\bf 1}, \qquad \{\sigma_j,\sigma_k\}=0 \quad \text{for } j \not=k. 
\end{equation}
The control $u$ in (\ref{Hamiltonian}) may represent a tuning of the interaction constant between the spins and the interaction of the spin with a  transverse  external field. The model can also be considered with two independent controls $u_0$ multiplying $H_0$ and $u_1$ multiplying $H_1$. The controllability analysis does not substantially change in that case and we chose the model (\ref{Hamiltonian}) as we have in mind adiabatic control protocols with $u \in [0,1]$ driving the state from an eigenvector of $H_1$ to an eigenvector of $H_0$.

The model (\ref{Hamiltonian}) is one of the most studied models in condensed matter physics especially in the setting of quantum phase transitions (see, e.g., \cite{Pierotti}, \cite{QPT}  )  
and goes under the name of {\it transverse field Ising model}. However, up to our knowledge, a control theoretic analysis of this model has not been  presented in the literature. We do this in this paper motivated by recent applications  of this model in quantum information processing and in particular as a testbed for a quantum annealing protocol \cite{DanielAN}. Such a protocol can, in fact, be seen as a {\it control protocol} where a `schedule' is decided (i.e., a function $u=u(t)$) in (\ref{Hamiltonian}) to drive the state from a (ground) eigenstate of $H_1$ to a (ground) eigenstate of $H_0$. Such a state may  encode  the solution of the computational task at hand. When $u=0$ the Hamiltonian in (\ref{Hamiltonian}) is $H_0$ and the ground state corresponds to spins that align  in opposite directions (anti-ferromagnetic)  or in the same direction (ferromagnetic) according to whether we do not place or we  place a minus sign in front of $H_0$. When $u=1$, the Hamiltonian reduces to $H_1$. This is the term that allows transitions between the eigenstates of $H_0$ and introduces disorder in the state of the system. From a quantum phase transition perspective there is a value of $u$ where the macroscopic behavior of the system (in the limit $n \rightarrow \infty$) will be drastically changed.

 The model (\ref{Hamiltonian}) is a special case of the {\it general} Ising spin model where one replaces the term $H_0$ with $H_0^{'}:=\sum_{j,k} J_{j,k} \sigma_z^j \sigma_z^k$ where $j,k$ refers to the spin location on a lattice and $J_{j,k}$ are coupling constant typically considered nonzero only for nearest neighbor spins. The general model is a fundamental conceptual model not only for phase transitions but also for phenomena in scientific areas different from condensed matter physics such as to describe bacterial vortexes in biology \cite{Bacter}.  The study of the controllability of these models is also motivated by the recent interest in (geometric) quantum machine learning, where the quantum evolution is seen as a learning protocols and existing symmetries in the Hamiltonians allow to keep unchanged the structure of the quantum data with the goal to improve robustness and to overcome some of the limitations of existing protocols. For a discussion of this point we refer to \cite{Marco2}, \cite{Marco1}, \cite{Scare}, and references therein (see, however,  the recent discussion in \cite{BP} for a comparison between quantum and classical machine learning in this context).    From a mathematics perspective, the recent paper \cite{Scare} is more related to the present one as it contains a complete classification of all the Lie algebras that can be generated by the terms appearing in 2-body Hamiltonians for 1-D structures such as (\ref{Hamiltonian}), (\ref{H0H1}). We shall however consider the Lie algebra generated by the `full' $H_0$ and $H_1$ Hamiltonians in (\ref{Hamiltonian}), (\ref{H0H1}).


From a quantum control theory point of view, system (\ref{Hamiltonian}) is an example of an {\it uncontrollable system}, that is, a system such that, independently of the control law applied, it is not possible to generate every  unitary operations. This is due to the presence of a symmetry group, that is,  a group of operators  that leave  the Hamiltonians describing the system invariant for any value of the control $u$. In particular, the Hamiltonian (\ref{Hamiltonian}) is invariant under the cyclic group $C_n$ generated by the permutation $(1 \, 2 \, \cdots \, n)$. Quantum control systems which admit a group of symmetries  have dynamics that split into the parallel of certain subsystems after an appropriate coordinate change. Techniques to find such a coordinate transformation, borrowing from Lie algebras and representation theory, have been explored in \cite{Mikobook} and \cite{conJonas}.  The state space splits into the direct sum of invariant subspaces and the system is said to be {\it subspace controllable} if full controllability is achieved on each (or some) of the invariant subspaces. Subspace controllability has been recently investigated in several quantum systems of interest, for example in \cite{Dan1}, \cite{Dan2},  and, in a different context,  in \cite{Marco1}.

The most popular procedure to test controllability of quantum systems (in the closed, not interacting with the environment,  case) involves the use of tools of Lie algebras and Lie group theory (see, e.g., \cite{Mikobook}). One considers the Hamiltonians available for the system (in our case $H_0$ and $H_1$ in (\ref{Hamiltonian})) which are Hermitian matrices. Once these are multiplied by the imaginary unit $i$, they become skew-Hermitian matrices, i.e., matrices in $u(N)$ ($su(N)$), the Lie algebra of $N \times N$ skew-Hermitian matrices (with zero trace). Here $N$ is the dimension of the system, with $N=2^{n}$ in our case. The Lie algebra ${\cal L}$ generated by these matrices is called the {\it dynamical Lie algebra} (DLA) and it determines the controllability properties of the system. In particular,  the set of reachable evolutions for the system is dense in the connected Lie group $e^{\cal L}$ associated with ${\cal L}$ and it is equal to $e^{\cal L}$ if $e^{\cal L}$ is compact.\footnote{This is most often the case when dealing  with quantum systems as the dynamical Lie algebra splits into the direct (commuting)  sum of an Abelian center ${\cal A}$  and a semisimple Lie algebra ${\cal S}$, i.e., ${\cal L}={\cal S} \oplus {\cal A}$ and $e^{\cal L}=e^{\cal S} e^{\cal A}=e^{\cal A} e ^{\cal S}$ and $e^{\cal S}$ is always compact (cf. , e.g., \cite{Mikobook}).} Therefore the DLA ${\cal L}$ characterizes the dynamical and control theoretic properties of a closed quantum system. We determine the dynamical Lie algebra ${\cal L}$ for the Ising spin periodic  chain (\ref{Hamiltonian}) in this paper, and  we analyze its structure in detail. We show that  it is the direct sum of a two dimensional center and $n-1$ simple Lie algebras isomorphic to $su(2)$.  We give the basis of each of these Lie algebras. 

In order to calculate the dynamical Lie algebra for the system (\ref{Hamiltonian}), (\ref{H0H1}), that is, the Lie algebra generated by $iH_0$ and $i H_1$, in principle,  we  have to perform (repeated) Lie brackets of such (large) matrices. In order to render the calculation tractable for any $n$,  we will describe  some computational techniques which are of more general interest. We will do this in section 
\ref{Meth}. In particular we will map a tensor product $i \sigma_1 \otimes \sigma_2 \otimes \cdots \otimes \sigma_n$ where $\sigma_j$, $j=1,...,n$ is one of the Pauli matrices  
$\{\sigma_{x,y,z} \}$ or the identity ${\bf 1}$ to the corresponding {\it Pauli string} (see, e.g., \cite{Reggio} and references therein)  $A_1A_2\cdots A_n$, with $A_j$ equal to $X,Y,Z,{\bf 1}$ according to whether $\sigma_j$ is $\sigma_x$, $\sigma_y$, $\sigma_z$, or ${\bf 1}$. The Lie bracket corresponds to a certain product between Pauli strings. In section \ref{CDLA}, we describe the dynamical Lie algebra by first giving  a basis of such a Lie algebra and then show how it splits in the direct (commuting) sum of a two dimensional center and $n-1$ simple ideals all of which isomorphic to $su(2)$. Section  \ref{LDC} is devoted to 
an example of application for a low dimensional case, $n=3$. We characterize the set of states that can be reached starting from the separable state $|000\rangle$ and prove that, although the dynamics can generate only a restricted  set of states,  it can lead to  states that have maximum {\it distributed entanglement} \cite{Wootters}.


\section{Methods}\label{Meth}

\subsection{Lie brackets of symmetric Hamiltonians}

Given two (or more) skew-Hermitian matrices $A$ and $B$ defining the dynamics of a quantum control system, a basis for the associated dynamical Lie algebra ${\cal L} $ is usually calculated  following an iterative algorithm of Lie brackets calculations described for example in \cite{Mikobook} (Chapter 3) (cf. the proof of Theorem \ref{DLAchar}). 
When considering Hamiltonians such as $H_0$ and $H_1$, which are sums of several terms but are  invariant under the action of a symmetry group, we need to organize the Lie bracket calculations  to make it feasible.  We describe here a method to do this which we will apply in the next section to compute a basis of the dynamical Lie algebra for our problem (\ref{Hamiltonian}). In our case the symmetry group can be taken as the cyclic group over $n$ elements $C_n$ generated by single translation $(1\, 2 \, \cdots \, n)$.\footnote{Notice that we use here the {\it cycle notation}  for a general element of the permutation group. Notice also that the Hamiltonians $H_0$ and $H_1$ in (\ref{H0H1}) are invariant under a larger group which includes `reflections' about the center element. For example in the case $n=5$ this includes the permutation $(1 \, 5)(2\, 4)$. For simplicity however we shall only consider permutation with respect to $C_n$.} 

Any element in a  Lie algebra that is invariant under the action of a finite group $G$ can be written as the {\it symmetrization} of an element $A$, that is,  as $\sum_{P \in G} P A P^{-1}$. 
For example, $iH_1$ in (\ref{H0H1}) can be written as $iH_1:=\sum_{P \in C_n} Pi \sigma_x^1 P^T$.  
Motivated by this, we define a {\it symmetrization operation} ${\cal C}(A):=\sum_{P \in G} PAP^{-1}$. 
When calculating the  Lie bracket of two symmetrized matrices, we have 
$$
\left[ {\cal C}(A),{\cal C}(B) \right]=\left[ \sum_{P \in G} PAP^{-1}, \sum_{Q \in G} QBQ^{-1} \right]=\sum_{P \in G} P \left( \sum_{Q \in G} [ A, P^{-1} QBQ^{-1}P ] \right) P^{-1}=
$$
$$
\sum_{P \in G} P \left( \sum_{S\in G} [A,SBS^{-1}]\right) P^{-1}
={\cal C} \left( \sum_{S\in G} \left[ A, SBS^{-1} \right] \right). 
$$
The above formula suggests to sum all the  commutators of a `{\it fixed}' $A$ with the elements in the `orbit' of $B$ under the action of the group $G$ and then take the symmetrization of the result. 

As an example,  let us calculate the Lie bracket $[i H_0, iH_1]$ which is the first step in the calculation of the dynamical Lie algebra for the system of $n$ spins in a circular chain. We have, 
$$
\left[ iH_1,iH_0\right] =\left[ {\cal C}(i\sigma_x^1), {\cal C}(i\sigma_z^1 \sigma_x^2)\right]={\cal C}\left(   \sum_{S\in C_n} [i\sigma_x^1, iS \sigma_z^1 \sigma_z^2S^T] \right). 
$$
When calculating $\sum_{S\in C_n} [i\sigma_x^1, iS \sigma_z^1 \sigma_z^2S^T]$, the only elements which give nonzero contributions are the Lie brackets $[i\sigma_x^1, i\sigma_z^1\sigma_z^2]$ and 
$[i\sigma_x^1, i\sigma_z^1\sigma_z^n]$, which give, respectively, using (\ref{commurel}) $-2i\sigma_y^1 \sigma_z^2$ and 
$-2i\sigma_y^1 \sigma_z^n$. This  gives  
\begin{equation}\label{firststep}
{\cal C}\left( -2i\sigma_y^1 \sigma_z^2 \right)+
{\cal C}\left( -2i\sigma_y^1 \sigma_z^n\right)=
{\cal C}\left( -2i\sigma_y^1 \sigma_z^2 \right)+
{\cal C}\left( -2i\sigma_z^1 \sigma_y^2 \right). 
\end{equation}
We shall not use calculations with `$\sigma$' Pauli matrices in the following but rather map such calculations into equivalent calculations with {\it Pauli strings}.

\subsection{Algebra with Pauli strings}

A  {\it Pauli string} (see, e.g. \cite{Reggio}) is a string of symbols in $\{X,Y,Z,{\bf 1}\}$ corresponding to a tensor product of Pauli matrices and $2 \times 2$ identities $\{\sigma_x,\sigma_y,\sigma_z, {\bf 1}\}$, multiplied by $i$. For example, $XYZ{\bf 1} {\bf 1} ZX{\bf 1}$ corresponds to the tensor product $i \sigma_x \otimes  \sigma_y \otimes \sigma_z \otimes {\bf 1} \otimes {\bf 1} \otimes \sigma_z \otimes \sigma_x \otimes {\bf 1}$. It is convenient to define  a product on the space of strings 
which corresponds to the Lie bracket of tensor products.  This is based on the commutation and anti-commutation relations (\ref{commurel}) 
and (\ref{anticommurel}) together with the relations 
\begin{equation}\label{tobeadded1}
[A\otimes B, C \otimes D]=\frac{1}{2} \{A,C\} \otimes [B,D] + \frac{1}{2} [A,C] \otimes \{B,D\},
\end{equation}
and 
\begin{equation}\label{anticomfor}
\{ A \otimes B, C \otimes D\}=\frac{1}{2} \{ A,C\} \otimes \{B,D\} + \frac{1}{2} [A,C] \otimes [B,D]. 
\end{equation}
In particular, when taking the product (commutator) of two strings $A:=A_1A_2\cdots A_n$, and 
$B:=B_1 B_2\cdots B_n$, one first looks at the places where $A$ and-or $B$ has an identity ${\bf 1}$. The corresponding product is in general a linear combination of strings which all have   in the same position the corresponding symbol of the other string. For example, when taking the product of $XYZ{\bf 1} {\bf 1} Z X {\bf 1} $ and ${\bf 1} Y Z X {\bf 1} X Z Z$, we will have 
$$
\begin{matrix} 
\qquad X & Y & Z & {\bf 1}  & {\bf 1} & Z & X & {\bf 1} \cr
\qquad {\bf 1} & Y & Z & X  & {\bf 1} & X & Z & {Z} \vspace{0.2cm} \cr
\hline \cr
\qquad X & ? & ?  & X  & {\bf 1} & ? & ? & Z
\end{matrix}, 
$$

Now all the other positions in both strings are occupied by a symbol in $\{X,Y,Z\}$. Then we consider positions where the two strings have the same symbol. These positions give a ${\bf 1}$ in the resulting string because of the relation (\ref{anticommurel}).  Therefore, in our example, we have 
$$
\begin{matrix} 
\qquad X & Y & Z & {\bf 1}  & {\bf 1} & Z & X & {\bf 1} \cr
\qquad {\bf 1} & Y & Z & X  & {\bf 1} & X & Z & {Z} \vspace{0.2cm} \cr
\hline \cr
\qquad X & {\bf 1} & {\bf 1}  & X  & {\bf 1} & ? & ? & Z
\end{matrix}, 
$$

We are left with the positions where the two strings have different symbols. If the number of these positions is even (including zero) the result is $0$.\footnote{The fact that commutators (anticommutators) of strings with an even (odd) numbers of different symbols in corresponding positions are zero is well known (see, e.g., \cite{Reggio} and references therein). We provide here justification to our statements for completeness.}  One can see this by induction on  $k$ for a number of positions $2k$. For $k=1$, we have zero from (\ref{tobeadded1}) since both anticommutators appearing on the right hand side are zero because of (\ref{anticommurel}). Using the same formula (\ref{tobeadded1}) and breaking two even  tensor products  as $A \otimes B$ and $C \otimes D$  with $B$ and $D$ a  tensor product with two factors (different for $B$ and $D$ in the corresponding positions) and using the  inductive assumption, one sees that zero is obtained for any even number. This case covers our example because the number of different positions between $XYZ{\bf 1} {\bf 1} ZX {\bf 1}$ and ${\bf 1} YZX{\bf 1} XZZ$ is $2$. Thus, the result is zero.

Consider now the case where the number of different positions is  {\it odd}. The result of $A_1A_2\cdots A_{2k+1} \, \times \, B_1B_2\cdots B_{2k+1}$, is $2(-1)^k  C_1 C_2\cdots C_{2k+1}$ with $C_j$ obtained according to the
 rules (\ref{commurel}), that is, $X \times Y \rightarrow Z$,   $Y \times Z \rightarrow X$, $Z \times X \rightarrow Y$ (with the sign $-$ if the order is inverted). To see 
this, we can use induction on $k$. For $k=0$ (single position) it is true,\footnote{Calculate for example $[\sigma_x, \sigma_y]=-2i\sigma_z$ from (\ref{commurel}). The two $i$'s in front of the tensor products give an extra $-1$ factor  which cancel the $-1$ in $-2i \sigma_z$, to give the desired result.} and assume it true for $k-1$, we can write the two tensor products as $E_1\otimes O_1$ and $E_2 \otimes O_2$ with $E_{1,2}$ strings of length $2$ and $O_{1,2}$ strings of length $2(k-1)+1$. We have, using (\ref{tobeadded1}) 
$$
\left[ E_1\otimes O_1, E_2 \otimes O_2 \right]=\frac{1}{2} \left( [E_1,E_2] \otimes \{ O_1, O_2\} + \{ E_1, E_2\} \otimes [O_1, O_2] \right) = 
\frac{1}{2}  \{ E_1, E_2\} \otimes [O_1, O_2],  
 $$
 since $[E_1,E_2]=0$. Furthermore write $E_1:=A \otimes B$ and $E_2=C \otimes D$.  Using (\ref{anticomfor}) and the fact that the anticommutators are zero in this case (from (\ref{anticommurel})), we get 
 $$
 \left[ E_1\otimes O_1, E_2 \otimes O_2 \right]=\frac{1}{4} [A,C] \otimes [B, D] \otimes [O_1, O_2]. 
 $$
 Now there is a factor $4$ that comes from the two commutators $[A,C]$ and $[B,D]$ according to (\ref{commurel}). Moreover, there is a factor $(-1)$ which comes from multiplying the Pauli matrices in the even products by $-i$. This, because of the inductive assumption applied to $  [O_1, O_2]$, proves the claim.


\vspace{0.25cm}

In the following when performing Lie brackets of element of the dynamical Lie algebra,  we shall use a combination of the techniques described in this section. In particular, let us, for the sake of illustration, calculate again the Lie bracket in (\ref{firststep}) for the case $n=3$ we fix $X{\bf 1} {\bf 1}$ and `circulate' $ZZ{\bf 1}$. Using the above rules we obtain 
$$
\begin{matrix} 
\qquad X & {\bf 1} & {\bf 1}  \cr
\qquad Z & Z  & {\bf 1} \vspace{0.2cm} \cr
\hline -2 \quad  Y & Z & {\bf 1} \cr
\qquad 
\end{matrix}, 
\qquad 
\begin{matrix} 
\qquad X & {\bf 1} & {\bf 1}  \cr
\qquad Z & {\bf 1}  & Z \vspace{0.2cm} \cr
\hline -2 \quad Y & {\bf 1} & Z \cr
\qquad 
\end{matrix}, 
\qquad 
\begin{matrix} 
\qquad X & {\bf 1} & {\bf 1}  \cr
\qquad {\bf 1} & Z  & Z \vspace{0.2cm} \cr
\hline \quad  & 0  & \quad   \cr
\qquad 
\end{matrix}, 
$$
which gives the result (\ref{firststep}) after we apply the symmetrizer ${\cal C}$.

\section{Characterization of the Dynamical Lie algebra}\label{CDLA}

Recall the definition of  the symmetrization operation ${\cal C}$ in the previous section for the case of  the cyclic group $C_n$,  
 ${\cal C}(A):=\sum_{P \in C_n} PAP^T$.  We define the following (linearly independent) skew-Hermitian operators 
\begin{equation}\label{Ys}
{\bf Y}^j:=i{\cal C}(\sigma_y\otimes \sigma_x^{\otimes j}\otimes \sigma_y \otimes {\bf 1}^{\otimes n-j-2}), \qquad j=0,1,...,n-2, 
\end{equation} 
\begin{equation}\label{Zs}
{\bf Z}^j:=i{\cal C}(\sigma_z\otimes \sigma_x^{\otimes j}\otimes \sigma_z\otimes {\bf 1}^{\otimes n-j-2}), \qquad j=0,1,...,n-2, 
\end{equation}
\begin{equation}\label{YZs}
{\bf YZ}^j:=i{\cal C}(\sigma_z\otimes \sigma_x^{\otimes j}\otimes \sigma_y\otimes {\bf 1}^{\otimes n-j-2})+i {\cal C}(\sigma_y\otimes \sigma_x^{\otimes j}\otimes \sigma_z\otimes {\bf 1}^{\otimes n-j-2}), \qquad j=0,1,...,n-2, 
\end{equation}
\begin{equation}\label{Xs}
{\bf X}:=i{\cal C}(\sigma_x\otimes {\bf 1}^{\otimes n-1}), \qquad {\bf XX}:=i{\cal C}(\sigma_x^{\otimes n-1} \otimes {\bf 1}). 
\end{equation}

\subsection{A basis of the dynamical Lie algebra} 

We start our analysis of the dynamical Lie algebra by calculating a basis  in the following theorem.    

\begin{theorem}\label{DLAchar} 
A basis of the dynamical Lie algebra ${\cal L}$ associated with a closed spin chain is given by the $3n-1$ skew-Hermitian matrices ${\bf Y}^j$, $j=0,1,...,n-2$, ${\bf Z}^j$, $j=0,1,...,n-2$, ${\bf YZ}^j$, $j=0,1,...,n-2$, ${\bf X}$ and ${\bf XX}$ defined in (\ref{Ys})-(\ref{Xs}).    
\end{theorem}

\begin{proof} The proof follows the iterative algorithm to obtain a basis of the dynamical Lie algebra which is described in Chapter 3 of \cite{Mikobook}. One starts with a set of generators, in this case ${\bf X}$ and ${\bf Z}^0$, which are considered elements at `depth 0'. At step $k=1,2,....$ one takes the Lie bracket of the elements of depth $k-1$ with the generators,  in this case ${\bf X}$ and ${\bf Z}^0$,  and eliminates all the linear combinations of elements of 
depth $\leq k-1$ to obtain the elements  of depth $k$. The process ends at a step $k$ when there are no new linearly independent elements and/or the number 
of linearly independent elements, which constitutes  the basis of the dynamical Lie algebra reaches $N^2$ or $N^2-1$ (the dimensions of $u(N)$ and $su(N)$ respectively), where $N$ is the dimension of the system, $2^n$ in our case. The algorithm ends in a finite number of steps since the system is finite dimensional. 

In our case, let us first consider the elements at depth $D=1,2,...,2n-3$. For these elements we  prove by induction on $D$ the following {\bf CLAIM}: The elements obtained for $D$ odd, $D=2k+1$, $k=0,1,...,n-2$,  are ${\bf YZ}^k$, and the elements obtained for $D$ even, $D:=2k$, $k=1,...,n-2$,  are ${\bf Y}^{k-1}$ and ${\bf Z}^k$. 

The case $D=1$ has been already done (modulo some shift in the notation) in (\ref{firststep}). Let us consider the case $D=2$. We calculate $[{\bf YZ}^0, {\bf X}]$, which gives, using Pauli strings algebra, 
$$
\begin{matrix} 
\qquad Z & Y & {\bf 1} & \cdots & {\bf 1}  \cr
\qquad X & {\bf 1}  & {\bf 1} & \cdots & {\bf 1} \vspace{0.2cm} \cr
\hline 2  \quad  Y & Y & {\bf 1} & \cdots & {\bf 1} \cr
\qquad 
\end{matrix}, \quad 
\begin{matrix} 
\qquad Z & Y & {\bf 1} & \cdots & {\bf 1}  \cr
\qquad {\bf 1} & X  & {\bf 1} & \cdots & {\bf 1} \vspace{0.2cm} \cr
\hline - 2  \quad  Z & Z & {\bf 1} & \cdots & {\bf 1} \cr
\qquad 
\end{matrix}, \qquad 
\begin{matrix} 
\qquad Z & Y & {\bf 1} & \cdots & {\bf 1}  \cr
\qquad {\bf 1} & {\bf 1}   & X & \cdots & {\bf 1} \vspace{0.2cm} \cr
\hline \quad   \quad  & \quad & 0  & \quad  & \quad \cr
\qquad 
\end{matrix}, \qquad 
$$
$$
\cdots \qquad \qquad 
\begin{matrix} 
\qquad Z & Y & {\bf 1} & \cdots & {\bf 1}  \cr
\qquad {\bf 1} & {\bf 1}   & {\bf 1} & \cdots & X \vspace{0.2cm} \cr
\hline \quad   \quad  & \quad & 0  & \quad  & \quad \cr
\qquad 
\end{matrix}, 
$$
and 
$$
\begin{matrix} 
\qquad Y & Z  & {\bf 1} & \cdots & {\bf 1}  \cr
\qquad X & {\bf 1}  & {\bf 1} & \cdots & {\bf 1} \vspace{0.2cm} \cr
\hline -2  \quad  Z & Z  & {\bf 1} & \cdots & {\bf 1} \cr
\qquad 
\end{matrix}, \quad 
\begin{matrix} 
\qquad Y & Z & {\bf 1} & \cdots & {\bf 1}  \cr
\qquad {\bf 1} & X  & {\bf 1} & \cdots & {\bf 1} \vspace{0.2cm} \cr
\hline 2  \quad  Y & Y & {\bf 1} & \cdots & {\bf 1} \cr
\qquad 
\end{matrix}, \qquad 
\begin{matrix} 
\qquad Y& Z & {\bf 1} & \cdots & {\bf 1}  \cr
\qquad {\bf 1} & {\bf 1}   & X & \cdots & {\bf 1} \vspace{0.2cm} \cr
\hline \quad   \quad  & \quad & 0  & \quad  & \quad \cr
\qquad 
\end{matrix}, \qquad 
$$
$$
\cdots \qquad \qquad 
\begin{matrix} 
\qquad Y & Z & {\bf 1} & \cdots & {\bf 1}  \cr
\qquad {\bf 1} & {\bf 1}   & {\bf 1} & \cdots & X \vspace{0.2cm} \cr
\hline \quad   \quad  & \quad & 0  & \quad  & \quad \cr
\qquad 
\end{matrix}, 
$$
which gives 
\begin{equation}\label{FRD2-A}
[{\bf YZ}^0, {\bf X}]=4{\bf Y}^0-4{\bf Z}^0. 
\end{equation}
Analogously we get 
\begin{equation}\label{FRD2-B}
[{\bf YZ}^0, {\bf Z}^0]=4{\bf X}+4{\bf Z}^1. 
\end{equation}
Eliminating ${\bf X}$ and ${\bf Z}^0$ which we already have from depth $0$ and normalizing an unimportant factor $4$, we obtain ${\bf Z}^1$ and ${\bf Y}^0$.  Notice that we have assumed that the depth $D$ is $\leq 2n-3$. Therefore such a calculation at depth $D=2$ is only of interest if $n>2$. 
The verification of the inductive steps follows similar calculations: Assume, by inductive assumption,  that at depth $D=2k$ even $k \geq 1$, we have $ {\bf Y}^{k-1}$ and 
${\bf Z}^k$. Then we calculate for depth $D=2k+1$, 
$$
\left[ {\bf Y}^{k-1},  {\bf X}\right]=-2 {\bf YZ}^{k-1}, \qquad \left[ {\bf Z}^k, {\bf X} \right]=2{\bf YZ}^k. 
$$
In these calculations we keep fixed the string corresponding to the first factor  and `move' $X$ in the string corresponding to the second factor ${\bf X}$. The only terms that give a nonzero contributions are the ones where $X$ corresponds to a $Y$ or $Z$ in the same position. 
For depth $D=2k+1$ we also calculate 
$$
\left[ {\bf Y}^{k-1},  {\bf Z}^0\right]=2 {\bf YZ}^{k}, \qquad \left[ {\bf Z}^k, {\bf Z}^0 \right]=-2{\bf YZ}^{k-1}. 
$$
Since ${\bf YZ}^{k-1}$ is achieved from the previous depth from inductive assumption,  the only new term is ${\bf YZ}^{k}$ as predicted by the {\bf CLAIM}  above. This concludes the verification of the inductive step of the {\bf CLAIM} when we go from an even depth to an odd one.

 Assume now we go from an odd depth to an even one. In particular, let the depth $D=2l+1< 2n-3$, and $D > 1$ (the case $l=0$, corresponding to $D=1$ has been already considered in the base step). We have ${\bf YZ}^l$. We calculate for depth $D=2l+2$, 
\begin{equation}\label{duecommutatori}
[{\bf YZ}^l, {\bf X}]=-4{\bf Z}^l+4{\bf Y}^l  \qquad [{\bf YZ}^l, {\bf Z}^0]=-4{\bf Y}^{l-1}+4{\bf Z}^{l+1}. 
\end{equation}
Setting $2k:=2l+2$, we have for the first commutator $-4{\bf Z}^{k-1}+ 4{\bf Y}^{k-1}$. By the inductive assumption ${\bf Z}^{k-1}$ was already obtained from the previous even depth. Therefore we obtain ${\bf Y}^{k-1}$. From the second commutator in (\ref{duecommutatori}) we obtain (since $l=k-1$) $-4{\bf Y}^{k-2}+4{\bf Z}^k$, which gives the new `direction' ${\bf Z}^k$ since $i{\bf Y}^{k-2}$ is already obtained from the previous even step.This agrees with the stated {\bf CLAIM}. 

Summarizing,  until depth $D= 2n-3$ we have the linearly independent (in fact orthogonal) elements ${\bf X}$, ${\bf Y}^j$, $j=0,1,...,n-3$, ${\bf Z}^j$, $j=0,1,...,n-2$, ${\bf YZ}^j$, $j=0,1,...,n-2$, having obtained $i{\bf YZ}^{n-2}$ at the last step. At depth $D=2n-2$, we obtain 
$$
\left[ {\bf YZ}^{n-2}, {\bf X} \right]= -4{\bf Z}^{n-2}+4{\bf Y}^{n-2}, 
$$ 
 which gives $ {\bf Y}^{n-2}$ since we have already obtained $ {\bf Z}^{n-2}$. 
 At depth $D=2n-2$, we also obtain $\left[ {\bf YZ}^{n-2}, {\bf Z}^0 \right]=4{\bf XX}-4{\bf Y}^{n-3}$, which gives ${\bf XX}$ since we have ${\bf Y}^{n-3}$ already. 
 
 Now the algorithm proceeds for depth $D=2n-1$. However we have 
 $[{\bf Y}^{n-2}, {\bf X}]=-2{\bf YZ}^{n-2}$, $[ {\bf Y}^{n-2}, {\bf Z}^0]=0$, $\left[{\bf XX}, {\bf X} \right]=0$, $[{\bf XX}, {\bf Z}^0 ]=-2{\bf YZ}^{n-2}$. None of these Lie brackets increases the dimension of the Lie algebra and therefore the algorithm stops. This concludes the proof of the theorem. 

\end{proof}

For convenience and later use we report in the following  table the commutators of the elements of the Lie algebras with the two generators ${\bf X}$ and ${\bf Z}^0$ as calculated in the proof of the theorem (using the Kronecker $\delta_{k,j}$ symbol). 

\begin{center}
 \vspace{0.5cm}

\begin{tabular}{ |c || c | c | }
\hline 
$[ \cdot, \cdot] $ & {\bf X} & ${\bf Z}^0$ \\
\hline   
$ {\bf Y}^k $ &    -2${\bf YZ}^k$  &  2 $(1-\delta_{k,n-2}){\bf YZ}^{k+1}$\\
\hline
${\bf Z}^k $&  $2{\bf YZ}^k$ & $-2(1-\delta_{0,k}) {\bf YZ}^{k-1}$\\ 
\hline
$ {\bf YZ}^k$ & $-4 {\bf Z}^k+4 {\bf Y}^k$ &     $4 \left( \delta_{k,0} {\bf X} - (1-\delta_{k,0}) {\bf Y}^{k-1} + \delta_{k,n-2} {\bf XX} + (1-\delta_{k,n-2}) {\bf Z}^{k+1} \right)$       \\
\hline 
$ {\bf X}$ &  0 & $-2 {\bf YZ}^0$ \\
\hline 
$ {\bf XX}$ & 0 & $-2{\bf YZ}^{n-2}$\\
\hline
 \end{tabular} 
 \vspace{0.5cm}

 {Table I}   
 
\end{center}

The above table contains all the needed information about the dynamical Lie algebra ${\cal L}$ since ${\bf X}$ and ${\bf Z}^0$ are its generators.

\subsection{Structure of the dynamical Lie algebra}

As a Lie subalgebra of $su(2^{n})$, the dynamical Lie algebra ${\cal L}$ is a {\it reductive} Lie algebra, namely the direct sum of an Abelian subalgebra, its {\it center}, and a certain number of simple ideals.  (cf., e.g.,  Chapter 3 and 4 in \cite{Mikobook}). To analyze the dynamical Lie algebra, therefore we start by calculating the center in the next subsection \ref{subseccen} and then we describe its simple ideals in the following subsection \ref{subsadd}. 

\subsubsection{The center of the dynamical Lie algebra}\label{subseccen}

We calculate  the center as the space of solutions of the system of linear equations in the variables $z_k, y_k, w_k, x,v$, $k=0,1,...,n-2$, 
$$
\left[ \sum_{k=0}^{n-2} z_k{\bf Z}^k+ \sum_{k=0}^{n-2} y_k {\bf Y}^k+ \sum_{k=0}^{n-2} w_k {\bf YZ}^k + x {\bf X} + v {\bf XX} \, , \, {\bf X} \right]=0, 
$$  
$$
\left[ \sum_{k=0}^{n-2} z_k{\bf Z}^k+ \sum_{k=0}^{n-2} y_k {\bf Y}^k+ \sum_{k=0}^{n-2} w_k {\bf YZ}^k + x {\bf X} + v {\bf XX} \, , \, {\bf Z}^0 \right]=0. 
$$  
It follows from the first equation using the commutators in Table I   
$$
\sum_{k=0} z_k [{\bf Z}^k, {\bf X}] + \sum_{k=0}^{n-2} y_k [{\bf Y}^k, {\bf X}] + \sum_{k=0}^{n-2} w_{k}[{\bf YZ}^k,{\bf X}]=  \sum_{k=0}^{n-2} 2 (z_k-y_k) {\bf YZ}^k+ \sum_{k=0}^{n-2} 4 w_k ({\bf Y}^k-{\bf Z}^k)=0, 
$$
which gives $w_k=0,$ and $y_k=z_k$ for $k=0,1,...,n-2$. Plugging these into the second equation and again using the relations in the table, we have
$$
\sum_{k=0}^{n-2} z_{k}\left( [{\bf Z}^k, {\bf Z}^0] + [ {\bf Y}^k, {\bf Z}^0] \right)+ x[{\bf X}, {\bf Z}^0]+v [{\bf XX}, {\bf Z}^0]= 
$$
$$2 \left( z_0 {\bf YZ}^1+\sum_{k=1}^{n-3} z_k (- {\bf YZ}^{k-1} + {\bf YZ}^{k+1} ) -z_{n-2}{\bf YZ}^{n-3} -x{\bf YZ}^0 -v{\bf YZ}^{n-2} \right) = 0,
$$ 
which gives $z_k=z_{k+2}$ for $k=0,1,..., n-4$, and $x=-z_1, v=z_{n-3}$. 
Depending on whether $n$ is even or odd, this yields two cases. If $n$ is odd, then $-x=z_1=z_3=...=z_{n-2}$ and $v=z_0=z_2=...=z_{n-3}$. If $n$ is even, then $-x=v=z_1=z_3=...=z_{n-3}$ and $z_0=z_2=...=z_{n-2}$. Therefore, we obtain the following:

\begin{proposition}\label{centerchar}
For \em{$n$ odd}, the center of the dynamical Lie algebra is spanned by 
\begin{equation}\label{centerodd}
C^o_1:=- {\bf X}+({\bf Y}^1+{\bf Z}^1)+ ({\bf Y}^3+{\bf Z}^3)+ \cdots + ({\bf Y}^{n-2}+{\bf Z}^{n-2}), 
\end{equation}
$$ C^o_2:= {\bf XX}+ ({\bf Y}^0+{\bf Z}^0)+ ({\bf Y}^2+{\bf Z}^2)+\cdots +  ({\bf Y}^{n-3}+{\bf Z}^{n-3}).
$$
For \em{$n$ even}, the center of the dynamical Lie algebra is spanned by 
\begin{equation}\label{centereven}
C^e_1:= - {\bf X}+{\bf XX}+ ({\bf Y}^1+{\bf Z}^1)+ ({\bf Y}^3+{\bf Z}^3)+ \cdots + ({\bf Y}^{n-3}+{\bf Z}^{n-3}),
\end{equation}
 $$ C^e_2:=({\bf Y}^0+{\bf Z}^0)+ ({\bf Y}^2+{\bf Z}^2)+\cdots +  ({\bf Y}^{n-2}+{\bf Z}^{n-2}).
$$

\end{proposition}

\subsubsection{The simple ideals of the dynamical Lie algebra}\label{subsadd}

To describe the simple ideals of the dynamical Lie algebra ${\cal L}$, we introduce an auxiliary sequence of polynomials $a_k=a_k(\lambda)$, for $k=-1,0,1,...,n-1$, defined recursively as 
\begin{equation}\label{recursiverel}
a_{-1}=0, \quad a_0=1, \quad a_k=\lambda a_{k-1}-a_{k-2}, \qquad \texttt{for } k=1,2,...,n-1.  
\end{equation}
An explicit expression  for $a_{k}=a_{k}(\lambda)$ is  given in the following lemma. 
\begin{lemma}\label{expexpr}
The polynomial $a_k=a_k(\lambda)$, for $k=1,2,...,n-1$ satisfying (\ref{recursiverel}) is given by 
\begin{equation}\label{EE1}
a_k(\lambda)=\sum_{j=0}^{\lfloor{\frac{k}{2}}\rfloor} (-1)^j \begin{pmatrix} k-j \cr j \end{pmatrix} \lambda^{k-2j}.
\end{equation}
\end{lemma}
We remark that the polynomials $a_k=a_k(\lambda)$ in (\ref{EE1}) have only odd (even) powers of $\lambda$ if $k$ is odd (even). The proof of this lemma, which is by induction on $k$ is given in the appendix. 

Another way to characterize the polynomials $a_{k}=a_{k}(\lambda)$ is given by the following lemma. For $k=1,...,n-1$, we denote by $A_k$ the $k\times k$ matrix which is $A_1:=\begin{pmatrix}{0}\end{pmatrix}$ and for $k=2,...,n-1$, $A_k$ is the tridiagonal $k \times k$ matrices with $1$'s above and below the main diagonal and zero everywhere else. For example $A_3:=\begin{pmatrix} 0 & 1 & 0 \cr 1 & 0 & 1 \cr 0 & 1 & 0 \end{pmatrix}$. Also denote by $I_k$ the $k \times k$ identity and by $p_B=p_B(\lambda)$ the characteristic polynomial of a matrix  $B$.   
\begin{lemma}\label{epiphany}
For $k=1,...,n-1$,  $a_k(\lambda)=\det{(\lambda I_k+ A_k)}=(-1)^k p_{-A}(\lambda)$. 
\end{lemma}
\begin{proof}
By induction on $k$, we see that the result is true for $k=1$ and $k=2$. Assuming the result true for $k-1$ and $k-2$, we calculate the determinant of 
$\lambda I_k+ A_k$ along the first column. This gives 
$$
\det (\lambda I_k +A_k)=\lambda \det(\lambda I_{k-1}+A_{k-1})-\det \left( \begin{pmatrix} 1 & 0 \qquad \cdots \qquad 0 \cr * & \lambda I_{k-2} + A_{k-2} \end{pmatrix} \right)=
\lambda a_{k-1}-a_{k-2}, 
$$
using the inductive assumption. This coincides with  the recursive definition (\ref{recursiverel}). 
\end{proof} 
In the following we shall be interested in the roots of the polynomial  $a_{n-1}=a_{n-1}(\lambda)$, i.e., (\ref{EE1}) with $k=n-1$ which are the eigenvalues of the matrix $-A_{n-1}$ discussed in the above lemma. We remark that  such roots/eigenvalues are all {\it real} since $A_{n-1}$ is symmetric.  Furthermore if $\lambda$ is an eigenvalue of $-A_{n-1}$, if we define $\vec a:=[a_0, a_1(\lambda),a_2(\lambda),...,a_{n-2}(\lambda)]^T$, the relations (\ref{recursiverel}) give $A_{n-1} \vec a=\lambda \vec a$, that is, $\lambda$ is an eigenvalue of $A_{n-1}$ also,  with eigenvector $\vec a$. The viceversa is also true since starting from the eigenvalue equation $A_{n-1} \vec a=\lambda \vec a$, we notice that the first component of $\vec a$ must be nonzero (otherwise the whole eigenvector $\vec a$ would be zero) and normalizing it to $1$ one finds that $\lambda$ must satisfy $a_{n-1}(\lambda)=0$, that is, it is an eigenvalue of $-A_{n-1}$. Since $A_{n-1}$ and $-A_{n-1}$ have the same eigenvalues, the eigenvalues of $A_{n-1}$ come in pairs $\pm \lambda$ except for the simple $0$ eigenvalue which occurs in the case where $n$ is even. The  algebraic multiplicity of each eigenvalue (the multiplicity of each root of $a_{n-1}$) is equal to the maximum number of linearly independent eigenvectors  since symmetric matrices are diagonalizable. However if $\vec a$ is an eigenvector of $A_{n-1}$ with eigenvalue $ \lambda$, the relations (\ref{recursiverel})  imply that two eigenvectors corresponding to the same eigenvalue $ \lambda$ have to be proportional to each-other. Thus $a_{n-1}=a_{n-1}(\lambda)$ has $n-1$ {\it distinct} roots. We shall not need the exact values of these roots to continue our theory. However an estimate will be useful. In particular,  we notice that for any root $\bar \lambda$,  
\begin{equation} \label{estimate3}
|\bar \lambda | <2.
\end{equation}
 To see this, notice that because of the symmetry about $0$, it is enough to show that $\bar \lambda < 2$. By an induction argument  if $\bar \lambda =2$ from (\ref{recursiverel}) it follows that $a_k=k+1$ and therefore $a_{n-1}=n\not=0$. If $\bar \lambda > 2$, then, again by induction we can show that $a_k$ is increasing with $k$, since $a_1=\bar \lambda > 1=a_0$ and 
$a_k=\bar \lambda a_{k-1}-a_{k-2} \leftrightarrow a_k-a_{k-1}=(\bar \lambda -1)a_{k-1} -a_{k-2}$ and  $(\bar \lambda -1)a_{k-1} -a_{k-2}>(\bar \lambda -1)a_{k-1} -a_{k-1}=(\bar \lambda -2)a_{k-1}>0$.

\vspace{0.25cm}

The reason for introducing the  polynomials $a_k$, $k=-1,0,1,...,n-1$ in the analysis of the dynamical Lie algebra ${\cal L}$ 
is explained in the following proposition. 
\begin{proposition}\label{keyprop}
Let $\bar \lambda$ be a real root  of the polynomial $a_{n-1}=a_{n-1}(\lambda)$, that is, (cf. Lemma \ref{epiphany}) an eigenvalue of the matrix $A_{n-1}$. Then the three elements 
\begin{equation}\label{Xhatted}
\hat X_{\bar \lambda}:= a_0 ({\bf X} + {\bf Z}^1)+ a_{n-2}(\bar \lambda) ( {\bf XX}-{\bf Y}^{n-3})+ \sum_{k=1}^{n-3} a_k(\bar \lambda) ({\bf Z}^{k+1}- {\bf Y}^{k-1}),  
\end{equation}
\begin{equation}\label{Yhatted}
\hat Y_{\bar \lambda} :=\sum_{k=0}^{n-2} a_k(\bar \lambda) \left( {\bf Y}^k-{\bf Z}^k \right),
\end{equation}
\begin{equation}\label{Zhatted}
\hat Z_{\bar \lambda}:=\sum_{k=0}^{n-2} a_k(\bar \lambda) {\bf YZ}^k, 
\end{equation}
span an \em{ideal} in the dynamical Lie algebra ${\cal L}$.\footnote{Notice that an easy verification shows that each of these ideal is orthogonal to the center described in Proposition \ref{centerchar}.} 
\end{proposition} 
In the following we shall denote by ${\cal I}_{\bar \lambda}$ the ideal corresponding to the root $\bar \lambda$ of $a_{n-1}$.
\begin{proof}
Since ${\bf X}$ and ${\bf Z}^0$ are generators of the Lie algebra ${\cal L}$, it is enough to show the two inclusions 
\begin{equation}\label{inclus2}
\left[ {\cal I}_{\bar \lambda}, {\bf X} \right] \subseteq {\cal I}_{\bar \lambda}, \qquad 
\left[ {\cal I}_{\bar \lambda}, {\bf Z}^0 \right] \subseteq {\cal I}_{\bar \lambda}. 
\end{equation}
To show the first one, we calculate, using Table I  (for simplicity of notation we omit the dependence on $\bar \lambda$)
$$
[\hat X, {\bf X}]=a_0 [{\bf Z}^1, {\bf X}]-a_{n-2}[{\bf Y}^{n-3}, {\bf X}]+ \sum_{k=1}^{n-3} a_k \left( [ {\bf Z}^{k+1}, {\bf X} ]- [{\bf Y}^{k-1}, {\bf X}] \right)= 
$$
$$
2a_0{\bf YZ}^1+2 a_{n-2} {\bf YZ}^{n-3}+2 \sum_{k=1}^{n-3} a_k ( {\bf YZ}^{k+1}+ {\bf YZ}^{k-1})=
$$
$$
2 a_1{\bf YZ}^0+\sum_{k=1}^{n-3} ( a_{k-1}+ a_{k+1}) {\bf YZ}^{k}+ 2 a_{n-3} {\bf YZ}^{n-2}= 2 \bar \lambda \hat Z,
$$
where in the last equality, we used (\ref{Zhatted}) (\ref{recursiverel}) and $a_{n-1}(\bar \lambda)=0$. Analogously, we obtain 
$$
[\hat Y, {\bf X}]=-4 \hat Z, \qquad [\hat Z, {\bf X}]=4 \hat Y. 
$$
For $[{\cal I}, {\bf Z}^0]$, we calculate 
$$
[\hat X, {\bf Z}^0]=a_0 \left( [ {\bf X}, {\bf Z}^0]+ [{\bf Z}^1, {\bf Z}^0] \right)+ 
a_{n-2} \left( [{\bf XX}, {\bf Z}^0]- [{\bf Y}^{n-3}, {\bf Z}^0]\right)+
\sum_{k=1}^{n-3} a_k \left( [{\bf Z}^{k+1}, {\bf Z}^0] - [{\bf Y}^{k-1}, {\bf Z}^0] \right)=
$$
$$
-4\sum_{k=0}^{n-2} a_k {\bf YZ}^k=-4 \hat Z.
$$
Also 
$$
[\hat { Y}, {\bf Z}^0]=\sum_{k=0}^{n-2} a_k \left( [{\bf Y}^k, {\bf Z}^0 - [{\bf Z}^k, {\bf Z}^0]\right)= 2 \sum_{k=0}^{n-3} a_k {\bf YZ}^{k+1} + 2 \sum_{k=1}^{n-2} a_k {\bf YZ}^{k-1}=
$$
$$
2 \left( a_1 {\bf YZ}^0+ \sum_{k=1}^{n-3} ( a_{k-1} + a_{k+1}) {\bf YZ}^k + a_{n-3} {\bf YZ}^{n-2} \right)= 2\bar \lambda \hat Z, 
$$
and 
$$
[\hat Z, {\bf Z}^0]=a_0 [{\bf YZ}^0, {\bf Z}^0] +\sum_{k=1}^{n-3} a_k [ {\bf YZ}^k, {\bf Z}^0] +a_{n-2} [{\bf YZ}^{n-2}, {\bf Z}^0]=
$$
$$
4 \left(a_0 ({\bf X}+{\bf Z}^1)+ \sum_{k=1}^{n-3} a_{k}({\bf Z}^{k+1}- {\bf Y}^{k-1} )+ a_{n-2} ({\bf XX}-{\bf Y}^{n-3}) \right)=4 \hat X. 
$$

\end{proof} 
The ideal ${\cal I}_{\bar \lambda}$ belongs to the semisimple subalgebra of ${\cal L}$. It is in fact semisimple since it 
must be reductive and if it had nontrivial center it would have a nonzero intersection with  the center which is 
however orthogonal to ${\cal I}_{\bar \lambda}$. By a dimension argument, it is also easily seen that ${\cal I}_{\lambda}$ is 
in fact {\it simple} since any one dimensional ideal would belong to the center. More in detail,  ${\cal I}_{\bar \lambda}$ is 
isomorphic to $su(2)$. Without resorting to the classification theory for Lie algebras (see, e.g., \cite{Ball}), we can explicitly display an isomorphism based on the above calculated Lie brackets.  {Such an isomorphism may be  useful in constructive control problems. For example, if one uses a control algorithm based on Lie groups decompositions (see, e.g., \cite{Mikobook}), it might be useful to know which matrix corresponds to $i \sigma_{x,y,z}$ for the Pauli matrices in \ref{PauliMat} (for example if using an Euler-type decomposition.) 

In order to obtain such an isomorphism,  we recall the inner product defined in $su(N)$ for general $N$, $\langle A, B\rangle:=\gamma Tr(AB^{\dagger})$ for a given $\gamma  > 0$. In order to have the matrices ${\bf X}$, ${\bf XX}$,  ${\bf Y}^j$, ${\bf Z}^j$, $j=0,...,n-2$ defined in (\ref{Ys}), (\ref{Zs}), (\ref{Xs}) with unit norm, we choose $\gamma=\frac{1}{\sqrt{n 2^n}}$. With this choice ${\bf YZ}^j$, $j=0,1,...,n-2$,  have norm $\sqrt{2}$, and the basis described in Theorem \ref{DLAchar} is  orthogonal. In the following for simplicity of notation, we omit the reference to $\bar \lambda$. Let us rewrite $\hat X$ and $\hat Y$ in (\ref{Xhatted}) and (\ref{Yhatted}) as  
\begin{equation}\label{Newformulas}
\hat X:=\hat A +\hat B, \quad \hat Y=\hat C+\hat D, 
\end{equation}
with $\hat A$, $\hat B$, $\hat C$, $\hat D$, defined as 
\begin{equation}\label{defABCD}
\hat A:= a_0{\bf X}+a_{n-2}{\bf XX}, \qquad  \hat C:=a_{n-2} {\bf Y}^{n-2}- a_0 {\bf Z}^0, 
\end{equation}
$$
\hat B:=\sum_{k=1}^{n-2} a_{k-1} {\bf Z}^k- \sum_{k=0}^{n-3} a_{k+1} {\bf Y}^k, \qquad  \hat D=\sum_{k=0}^{n-3} a_k {\bf Y}^k-\sum_{k=1}^{n-2} a_k {\bf Z}^k.
$$
We notice that all these matrices are orthogonal to each other except for $\hat B$ and $\hat D$ which are such that $\langle \hat B, \hat D \rangle=-2\sum_{k=1}^{n-2} a_k a_{k-1}$. Furthermore $\|\hat A\|=\|\hat C\|$ and $\|\hat B\|=\|\hat D\|$. In these  cases an orthogonal basis for $\texttt{span} \{\hat X, \hat Y\}$ can be obtained as $\{ \hat X+\hat Y,  \hat X -\hat Y\}$. Therefore we define two new matrices in the ideal, which are a basis of $\texttt{span} \{ \hat X, \hat Y\}$ as 
\begin{equation}\label{tildasX}
\tilde S_x=\hat X+\hat Y=(\hat A + \hat C) + (\hat B + \hat D)= a_0 ({\bf X} -{\bf Z}^0)+ a_{n-2}({\bf XX} +{\bf Y}^{n-2})+ \sum_{k=1}^{n-2}(a_{k-1}-a_{k}) {\bf Z}^k + \sum_{k=0}^{n-3} (a_k-a_{k+1}) {\bf Y}^k
\end{equation}
\begin{equation}\label{tildasY}
\tilde S_y=\hat X - \hat Y= (\hat A-\hat C)+(\hat B - \hat D)=a_0({\bf X}+{\bf Z}^0)+a_{n-2}({\bf XX}-{\bf Y}^{n-2})+ \sum_{k=1}^{n-2}(a_{k-1}+a_k) {\bf Z}^k -\sum_{k=0}^{n-3} (a_{k+1}+a_{k}){\bf Y}^k. 
\end{equation}
It is useful to calculate the norms of $\tilde S_x$ and $\tilde S_y$ as $\| \tilde S_x\|^2=2\sum_{k=0}^{n-1}(a_k-a_{k-1})^2$,
 $\|\tilde S_y\|^2=2\sum_{k=0}^{n-1}(a_k+a_{k-1})^2$ (recall $a_{-1}=a_{n-1}=0$). 
With these definitions, we have the isomorphism described in the following proposition whose  proof is given in the appendix.
\begin{proposition}\label{su2} Denote by $\bar \lambda$ a root of the polynomial $a_{n-1}=a_{n-1}(\lambda)$. Define the basis of ${\cal I}_{\bar \lambda}$, 
\begin{equation}\label{SxSySz7}
S_x:=-\frac{1}{2\sqrt{2-\bar \lambda} \|\tilde S_x \| \|\tilde S_y\|} \tilde S_x,  \quad 
S_y:=\frac{1}{2\sqrt{2- \bar \lambda}\|\tilde S_y\|^2} \tilde S_y, \qquad  S_z:=\frac{1}{2 \|\tilde S_x\|\|\tilde S_y\|} \hat Z.
\end{equation}
Then $S_{x,y,z}$  satisfy the commutation relation 
\begin{equation}\label{commureln}
[S_x,S_y]=S_z, \qquad [S_y, S_z]=S_x, \qquad [S_z,S_x]=S_y. 
\end{equation}
Therefore  the ideal ${\cal I}_{\bar \lambda}$ is isomorphic to $su(2)$. 
\end{proposition}

}

Consider now two {\it different} roots of the polynomial $a_{n-1}=a_{n-1}(\lambda)$, $\bar \lambda_1$ and $\bar \lambda_2$. The two associated ideals ${\cal I}_{\bar \lambda_1}$ and ${\cal I}_{\bar \lambda_2}$ have zero intersection or they coincide (because an intersection  of smaller dimension would mean that there is a center of dimension higher than two). However, if  ${\cal I}_{\bar \lambda_1}={\cal I}_{\bar \lambda_2}$ we should have 
$\hat Z_{\bar \lambda_1}=k \hat Z_{\bar \lambda_2}$ for $k\not=0$ which implies, by comparing the first two entries in (\ref{Zhatted}) and using $a_0$ and $a_1$ from (\ref{recursiverel}),  that $\bar \lambda_1=\bar \lambda_2$. Therefore distinct real roots correspond to  distinct ideals. Thus, we have
\begin{proposition}\label{distinct}
Distinct real roots $\bar \lambda_1$ and $\bar \lambda_2$ of the polynomial $a_{n-1}=a_{n-1}(\lambda)$ defined in (\ref{EE1}) 
 correspond to two disjoint  simple ideals each isomorphic to $su(2)$, spanned by $\hat X$, $\hat Y$ and $\hat Z$ defined in (\ref{Xhatted}), (\ref{Yhatted}) and (\ref{Zhatted}) (or $S_x,$ $S_y$, $S_z$ defined in (\ref{SxSySz7})), with $\bar \lambda=\bar \lambda_1$ or  $\bar \lambda=\bar \lambda_2$, respectively. The dynamical Lie algebra ${\cal L}$ contains therefore $n-1$ simple ideals, each corresponding to one root of $a_{n-1}$.  
\end{proposition}

In conclusion we can summarize the main results of this section in the following theorem which describes the dynamical Lie algebra of the system (\ref{Hamiltonian}), (\ref{H0H1}). 

\begin{theorem}\label{MainSum}
The dynamical Lie algebra ${\cal L}$ of the quantum control system  (\ref{Hamiltonian}), (\ref{H0H1}) has a basis given by $\{ {\bf Y}^k,{\bf Z}^k, {\bf YZ}^k, {\bf X}, {\bf XX}  \, | \, k=0,...,n-2\}$, defined in (\ref{Ys}), (\ref{Zs}), (\ref{YZs}), (\ref{Xs}). It is the direct sum of a two dimensional center spanned by the matrices in (\ref{centerodd}) if $n$ is odd or (\ref{centereven}) if $n$ is even and $n-1$ simple ideals isomorphic to $su(2)$. Such ideals are parametrized by the roots of the polynomial $a_{n-1}=a_{n-1}(\lambda)$ which coincide with the eigenvalues of the matrix $A_{n-1}$ defined in Lemma \ref{epiphany}. For a fixed root $\bar \lambda$ a basis of the associated ideal is given by $\{ \hat X_{\bar \lambda}, \hat Y_{\bar \lambda}, \hat Z_{\bar \lambda}\}$ defined in (\ref{Xhatted}), (\ref{Yhatted}), (\ref{Zhatted}).   
\end{theorem}

\section{ Example}\label{LDC}

We conclude the paper with an example of application of its main controllability result. We consider the Ising spin chain (\ref{Hamiltonian}), (\ref{H0H1}) with $n=3$ and ask the question of the possible states that can be reached starting from $|000\rangle$. Since the initial state is fully separable, we shall also ask the question of the maximum {\it distributed entanglement}  that can be generated starting from this state. We shall use as  a measure of the distributed entanglement the {\it tangle} among the three qubits as defined in \cite{Wootters} and elaborated upon in \cite{QIC}. 

Specializing the results to the case $n=3$, we find that, the dynamical Lie algebra ${\cal L}$ associated with the Ising spin chain  has a basis $\{ {\bf Y}^0, {\bf Y}^1, {\bf Z}^0, {\bf Z}^1, {\bf YZ}^0, {\bf YZ}^1, {\bf X}, {\bf XX} \}$, and center spanned by $\{ {\bf Y}^0 + {\bf Z}^0 + {\bf XX}, {\bf Y}^1 + {\bf Z}^1 -{\bf X} \}$. The corresponding polynomial $a_2(\lambda) = \lambda^2-1$ has the roots $\lambda=\pm1$, and, therefore, the simple ideals isomorphic to $su(2)$, ${\cal I}_{1}$ and ${\cal I}_{-1}$, are the ones generated by ${\bf YZ}^0 + {\bf YZ}^1$ and ${\bf YZ}^0 -{\bf YZ}^1$. They   have bases
\begin{equation}\label{basis1}
\{ {\bf Y}^1 -{\bf Z}^0 + {\bf X} + {\bf XX}, -2{\bf Y}^0- {\bf Y}^1 + {\bf Z}^0 + 2{\bf Z}^1 + {\bf X} + {\bf XX},  {\bf YZ}^0 + {\bf YZ}^1 \},
\end{equation}
\begin{equation}\label{basis2}
\{ 2{\bf Y}^0- {\bf Y}^1 - {\bf Z}^0 + 2{\bf Z}^1 + {\bf X} - {\bf XX}, {\bf Y}^1 + {\bf Z}^0 + {\bf X} - {\bf XX},  {\bf YZ}^0 - {\bf YZ}^1 \}
\end{equation}
respectively. Now we will make use of the fact that the invariant subspaces of $u^{C_n}(2^n)$ are also invariant for the Lie algebra under consideration.\footnote{We use the notation $u^G(N)$ for the subalgebra of $u(N)$ that commutes with the group $G$. Since ${\cal L} \subseteq u^{C_n}(2^n)$, invariance with respect to $u^{C_n}(2^n)$ implies invariance with respect to ${\cal L}$.} In the case $n=3$, the invariant subspaces  of $u^{C_n}(2^n)$ are discussed in the paper \cite{conJonas}. In particular, our initial state $|000\rangle$ is in the invariant subspace spanned by the permutation invariant states,  also called the {\it Dicke states}, 
\begin{equation}\label{Dicke}
|\phi_0\rangle:=|000\rangle,  \, |\phi_1\rangle :=|111\rangle, \,  |\phi_2\rangle:=\frac{1}{\sqrt{3}}\left( |100\rangle+|010\rangle+|001\rangle \right), \, |\phi_3\rangle:=\frac{1}{\sqrt{3}}\left(|011\rangle+|101\rangle+|110\rangle \right),
\end{equation}
 Therefore, it suffices to consider the $4$-dimensional invariant subspace spanned by the basis in (\ref{Dicke}).  By calculating the action of the basis elements of ${\cal I}_{1}$ in (\ref{basis1}) on the basis $(\ref{Dicke})$, it can be  seen that each element has the form $\sigma \otimes A$, where $\sigma$ may be any  element in $su(2)$ and $A=\begin{pmatrix} 1 & 1  \cr  1 & 1 \end{pmatrix}$. Similarly, the action of the basis elements of ${\cal I}_{-1}$ in (\ref{basis2})  have the form $\sigma \otimes B$, where $B=\begin{pmatrix} 1 & -1  \cr  -1 & 1 \end{pmatrix}$. Furthermore, two basis elements of the center take  the form of $-3i {\bf 1} \otimes \sigma_x$ and $3i {\bf 1} \otimes {\bf 1}$. 
The matrix $A$ has the eigenvalues $2,0$ with the eigenvectors $\vec v_2= \begin{pmatrix} 1  \cr  1 \end{pmatrix}$ and $\vec v_0= \begin{pmatrix} 1  \cr  -1 \end{pmatrix}$, respectively, while  $B$ has the same eigenvalues with the swapped eigenvectors $\vec v_0 = \begin{pmatrix} 1  \cr  1 \end{pmatrix}$ and $\vec v_2 = \begin{pmatrix} 1  \cr  -1 \end{pmatrix}$. This implies that the 4-dimensional invariant subspace can be split  into the direct sum of  two invariant subspaces 
$V_1= \texttt{span} \left\{ v \otimes \begin{pmatrix} 1 \cr 1 \end{pmatrix} \,|\, v \in \mathbb{C}^2 \right\} $ and $V_2= \texttt{span} \left\{ v \otimes \begin{pmatrix} 1 \cr -1 \end{pmatrix} \, | \, v \in \mathbb{C}^2 \right\} $.

Writing the initial state according to its components in $V_1$ and $V_2$ as  
$|000\rangle = \frac{1}{2}\begin{pmatrix} 1  \cr  0 \end{pmatrix} \otimes \begin{pmatrix} 1  \cr  1 \end{pmatrix} + \frac{1}{2}\begin{pmatrix} 1  \cr  0 \end{pmatrix} \otimes \begin{pmatrix} 1  \cr  -1 \end{pmatrix}$,  we can act with arbitrary matrices in $SU(2)$ on the first components in the tensor products appearing in this decomposition and we can add  an arbitrary phase difference between the two components. This leads to  the set of  reachable states from $|000\rangle$, in the basis of the normalized Dicke states, (neglecting a common unphysical, phase factor)  
\begin{equation}\label{insieme}
{\cal O}:=\left\{ \vec v \in \mathbb{C}^4 \,\Big| \, \vec v=  \frac{e^{i\mu}}{2}  \begin{pmatrix} e^{i\phi} \cos(\theta) \cr e^{i\phi} \cos(\theta) \cr e^{i\zeta} \sin(\theta) \cr e^{i\zeta} \sin(\theta) \end{pmatrix} + \frac{e^{-i\mu}}{2} \begin{pmatrix} e^{i \alpha} \cos(\gamma) \cr - e^{i \alpha} \cos(\gamma) \cr e^{i \beta} \sin(\gamma) \cr - e^{i\beta} \sin(\gamma) \end{pmatrix},  \,  \,  \phi, \zeta, \psi, \alpha, \beta, \gamma, \mu \in \mathbb{R} \right\}.
\end{equation}
The set ${\cal O}$ in (\ref{insieme}) shows that not all the states in the symmetric sector  can be achieved with our control system because if we write the Schmidt decomposition (see, e.g., \cite{NC}) of a vector $\psi=r_1 \vec e_1 \otimes \vec f_1+ r_1 \vec e_2 \otimes \vec f_2$ with $\{ \vec e_1, \vec e_2\}$ and $\{ \vec f_1, \vec f_2\}$ orthonormal bases of $\mathbb{C}^2$, we are constrained to have $r_1=r_2$. We shall show however that the class of states in (\ref{insieme}) can achieve maximum distributed entanglement. Consider  the family of states ${\cal {F}}:=\{ \vec v=\begin{pmatrix}\cos(\theta), \,  0, \,    \sin(\theta), \,    0  \end{pmatrix}^T \,|\, \theta \in \mathbb{R}\}$, which is a subset of ${\cal O}$ in (\ref{insieme}) by setting $\alpha=\phi=\zeta=\beta=\mu=0$ and $\gamma=\theta$.  By applying the formula for the tangle on the symmetric sector \cite{QIC} (which simplifies in the case where the last component of the vector is zero), we  obtain the tangle $\tau$ as a function of $\theta$. In particular, we have
$$
\tau=\frac{16}{3 \sqrt{3}} |\cos(\theta)\sin^3(\theta)|. 
$$
It is a calculus exercise to show that the maximum of this function is obtained when $|\sin(\theta)|=\frac{\sqrt{3}}{2}$ and $|\cos(\theta)|=\frac{1}{2}$. This maximum is equal to $1$. Therefore we have that the dynamics of the system (\ref{Hamiltonian}), (\ref{H0H1}) can induce maximum distributed entanglement, and in fact, it does that starting from the separable state $|000\rangle$. This and other quantum information theoretic properties are of interest for general $n$ with the allowed dynamics described in this paper.

\section*{Acknowledgement} This research was supported by the ARO MURI grant W911NF-22-S-0007. The authors wish 
to thank Prof. D. Lidar for helpful comments and suggestions. They also would like to thank Prof. T. Iadecola for a useful 
discussion on quantum phase transitions and for pointing out reference \cite{Scare}.

\begin{appendix}

\section{Proof of Lemma \ref{expexpr}}

\begin{proof}
For $k=1,2$,  it is directly verified that (\ref{recursiverel}) and (\ref{EE1}) give the same result. For $k >2$ we   calculate using (\ref{EE1}) 
$$
\lambda a_{k-1}-a_{k-2}= \sum_{j=0}^{\lfloor{\frac{k-1}{2}}\rfloor} (-1)^j \begin{pmatrix}  k-1-j\cr j\end{pmatrix} \lambda^{k-2j}- 
 \sum_{j=0}^{\lfloor{\frac{k-2}{2}}\rfloor} (-1)^j \begin{pmatrix}  k-2-j\cr j\end{pmatrix} \lambda^{k-2-2j}=
$$ 
$$
\sum_{j=0}^{\lfloor{\frac{k-1}{2}}\rfloor} (-1)^j \begin{pmatrix}  k-1-j\cr j\end{pmatrix} \lambda^{k-2j}+
 \sum_{j=1}^{\lfloor{\frac{k-2}{2}}\rfloor+1} (-1)^j \begin{pmatrix}  k-1-j\cr j-1\end{pmatrix} \lambda^{k-2j}
$$
Assume that $k$ is even. Then we have 
$$
\lambda a_{k-1}-a_{k-2}=\sum_{j=0}^{\frac{k-2}{2}} (-1)^j \begin{pmatrix} k-1-j \cr j \end{pmatrix} \lambda^{k-2j} +\sum_{j=1}^{\frac{k}{2}} (-1)^j \begin{pmatrix}k-1-j \cr j-1 \end{pmatrix}\lambda^{k-2j}=
$$
$$
\lambda^k+(-1)^{\frac{k}{2}}+ \sum_{j=1}^{\frac{k}{2}-1} \left\{\begin{pmatrix} k-1-j \cr j \end{pmatrix}+ \begin{pmatrix} k-1-j \cr j-1 \end{pmatrix} \right\} \lambda^{k-2j}, 
$$
which coincides with (\ref{EE1}).\footnote{Because of the relation $\begin{pmatrix} k-1-j \cr j \end{pmatrix}+ \begin{pmatrix} k-1-j \cr j-1 \end{pmatrix}=\begin{pmatrix} k-j \cr j\end{pmatrix}$.}

Now assume $k$ is odd. Then $\lfloor \frac{k-1}{2} \rfloor=\lfloor \frac{k-2}{2} \rfloor+1 =\lfloor \frac{k}{2} \rfloor$, and using the fact that 
$\begin{pmatrix} k-1 \cr 0  \end{pmatrix}=\begin{pmatrix} k \cr 0  \end{pmatrix}=1$, the sum becomes 
$$
\lambda a_{k-1}-a_{k-2}= \lambda^k+ \sum_{j=1}^{\lfloor \frac{k}{2} \rfloor} (-1)^j  \left\{ \begin{pmatrix} k-1-j \cr j \end{pmatrix}+\begin{pmatrix} k-1-j \cr j-1 \end{pmatrix}  \right\}\lambda^{k-2j}= \lambda^k+  \sum_{j=1}^{\lfloor \frac{k}{2} \rfloor} (-1)^j \begin{pmatrix} k- j \cr j \end{pmatrix} \lambda^{k-2j}, 
$$
which coincides with (\ref{EE1}).

\end{proof}

\section{Proof of Proposition \ref{su2}}

\begin{proof}
 From (\ref{tildasX}) we can write 
 $({\bf X}-{\bf Z}^0)=
 \frac{2a_0}{\| \tilde S_x\|^2} \tilde S_x+ 
 ({\bf X}-{\bf Z}^0)_{\perp}=
 \frac{a_0}{\sum_{k=0}^{n-1}(a_k-a_{k-1})^2}\tilde S_x+ ({\bf X}-{\bf Z}^0)_{\perp}$ 
where 
$(A)_{\perp}$  denotes the component of $A$ in the orthogonal complement of the ideal ${\cal I}$ (which commutes with ${\cal I}$), and 
analogously  from (\ref{tildasY})   $({\bf X}+{\bf Z}^0)=\frac{2a_0}{\|\tilde S_y\|^2} \tilde S_y+ ({\bf X}+{\bf Z}^0)_{\perp}
=\frac{a_0}{\sum_{k=0}^{n-1}(a_k+a_{k-1})^2} \tilde S_y+ 
({\bf X}+{\bf Z}^0)_{\perp}$. These formulas are useful when calculating the commutators among $\tilde S_x$, $\tilde S_y$ and $\hat Z$ without using  all the Lie brackets between the elements of the basis in Theorem \ref{DLAchar}.

We shall prove that 
\begin{equation}\label{commureln9}
[\tilde S_x, \tilde S_y]=-(4-2 \bar \lambda)\|\tilde S_y\|^2 \hat Z, \qquad [\tilde S_y, \hat Z]=-2 \|\tilde S_y\|^2 \tilde S_x, \qquad [\hat Z, \tilde S_x]=-2 \|\tilde S_x\|^2 \tilde S_y, 
\end{equation}
from which, using (\ref{SxSySz7}), relations (\ref{commureln}) follow. Write $\tilde S_y=\frac{\|\tilde S_y\|^2}{2}({\bf X}+{\bf Z}^0)-\frac{\|\tilde S_y\|^2}{2}({\bf X}+{\bf Z}^0)_{\perp}$, 
so that we have $[\tilde S_x,\tilde S_y]=\frac{\|\tilde S_y\|^2}{2} [\tilde S_x, {\bf X}+{\bf Z}^0]$. Using (\ref{tildasX}) and the commutation table with the generators, Table I, 
we obtain 
$$
[\tilde S_x,\tilde S_y]=-\|\tilde S_y\|^2 \left(2 a_0 {\bf YZ}^0+2a_{n-2}{\bf YZ}^{n-2}+2 \sum_{k=1}^{n-2}(a_k-a_{k-1}){\bf YZ}^k+2 \sum_{k=0}^{n-3}(a_k-a_{k+1}){\bf YZ}^k \right)=
$$
$$
-\|\tilde S_y\|^2\left(   (4a_0-2a_1) {\bf YZ}^0+ (4a_{n-2}-2a_{n-3}){\bf YZ}^{n-2} +2\sum_{k=1}^{n-3} (2a_k-a_{k+1} -a_{k-1}) {\bf YZ}^k \right)=
$$
$$
-\|\tilde S_y\|^2\left( 4-2\bar \lambda \right) \left( \sum_{k=0}^{n-2} a_k {\bf YZ}^k\right)= -(4-2 \bar \lambda)\|\tilde S_y\|^2 \hat Z,
$$
where we used $a_1=\bar \lambda a_0$, $a_{k-1}+a_{k+1}=\bar \lambda a_k$ (from (\ref{recursiverel}))  and $a_{n-3}=\bar \lambda a_{n-2}$, from the fact that $a_{n-1}=\bar \lambda a_{n-2}-a_{n-3}=0$, along with the definition (\ref{Zhatted}). Similar calculations lead to the second and third one in (\ref{commureln9}).  
\end{proof}

\end{appendix}

\end{document}